\documentclass[pra,floatfix,twocolumn,superscriptaddress,showpacs,preprintnumbers,longbibliography,footinbib]{revtex4-1}
\usepackage{dcolumn}    
\usepackage{bm} 
\usepackage{graphicx}
\usepackage{lipsum}
\usepackage{amsmath}  
\usepackage{amsthm}
\usepackage{enumitem}
\usepackage{latexsym}
\usepackage{amsfonts}   
\usepackage{amssymb}
\usepackage{array}      
\usepackage{epsfig}
\usepackage{braket} 
\usepackage{bbold}
\usepackage{color}
\usepackage{soul}
\usepackage[colorlinks=true,linkcolor=blue,urlcolor=blue,citecolor=blue,pdfusetitle]{hyperref}
\usepackage{hyperref}

\usepackage{notes2bib}

\theoremstyle{definition}
\newtheorem{definition}{Definition}[section]
\theoremstyle{plain}
\newtheorem{theorem}{Theorem}[section]
\newtheorem{lemma}{Lemma}[section]
\newtheorem{corollary}{Corollary}[section]

\newcommand{\ra}{\rangle}
\newcommand{\la}{\langle}

\newcommand{\sgn}{ \textrm{sgn}}

\newcommand{\eq}{Eq.~}

\newcommand{\fig}{Fig.~}

\newcommand{\Ao}[1]{\hat{a}_{#1}}
\newcommand{\Aod}[1]{\hat{a}^\dag_{#1}}

\DeclareRobustCommand\openzero{\leavevmode\hbox{0\kern-.55em0}}

\usepackage{xcolor}

\begin{document}

\begin{abstract}
Thermal machines are physical systems designed to convert thermal energy into practical work through cyclic state transformations. A key component in such a machine is a clock-equipped control element that dictates which interaction Hamiltonian governs the system-reservoir interactions at specific times, while itself remaining unaffected. However, in the context of quantum dynamics, it is well known that maintaining perfect isolation is practically impossible, except under highly idealized conditions. In this study, we begin with such an idealized model for a clock and systematically relax its main assumptions to develop a more realistic framework. Our approach yields a simplified yet physically consistent description of clock dynamics, enabling the analysis of deviations from ideal time-keeping, which we interpret as clock degradation. We introduce a continuous time operator to derive a lower bound on this degradation, grounded in a generalized time-energy uncertainty relation. These results highlight trade-offs between clock performance and energy uncertainty in physically realizable control schemes.
\end{abstract}

\title{Physically constrained quantum clock-driven dynamics}

\author{Lea Lautenbacher}
\email{lea.lautenbacher@uni-ulm.de}
\affiliation{Institute of Theoretical Physics \& IQST, Ulm University, Albert-Einstein-Allee 11 89081, Ulm, Germany}

\author{Giovanni Spaventa}
\affiliation{Institute of Theoretical Physics \& IQST, Ulm University, Albert-Einstein-Allee 11 89081, Ulm, Germany}

\author{Dario Cilluffo}
\email{dario.cilluffo@uni-ulm.de}
\affiliation{Institute of Theoretical Physics \& IQST, Ulm University, Albert-Einstein-Allee 11 89081, Ulm, Germany}

\author{Susana F. Huelga}
\affiliation{Institute of Theoretical Physics \& IQST, Ulm University, Albert-Einstein-Allee 11 89081, Ulm, Germany}

\author{Martin B. Plenio}
\affiliation{Institute of Theoretical Physics \& IQST, Ulm University, Albert-Einstein-Allee 11 89081, Ulm, Germany}

\maketitle

\section{Introduction}

\begin{figure}[]
\includegraphics[scale=0.35,angle=0]{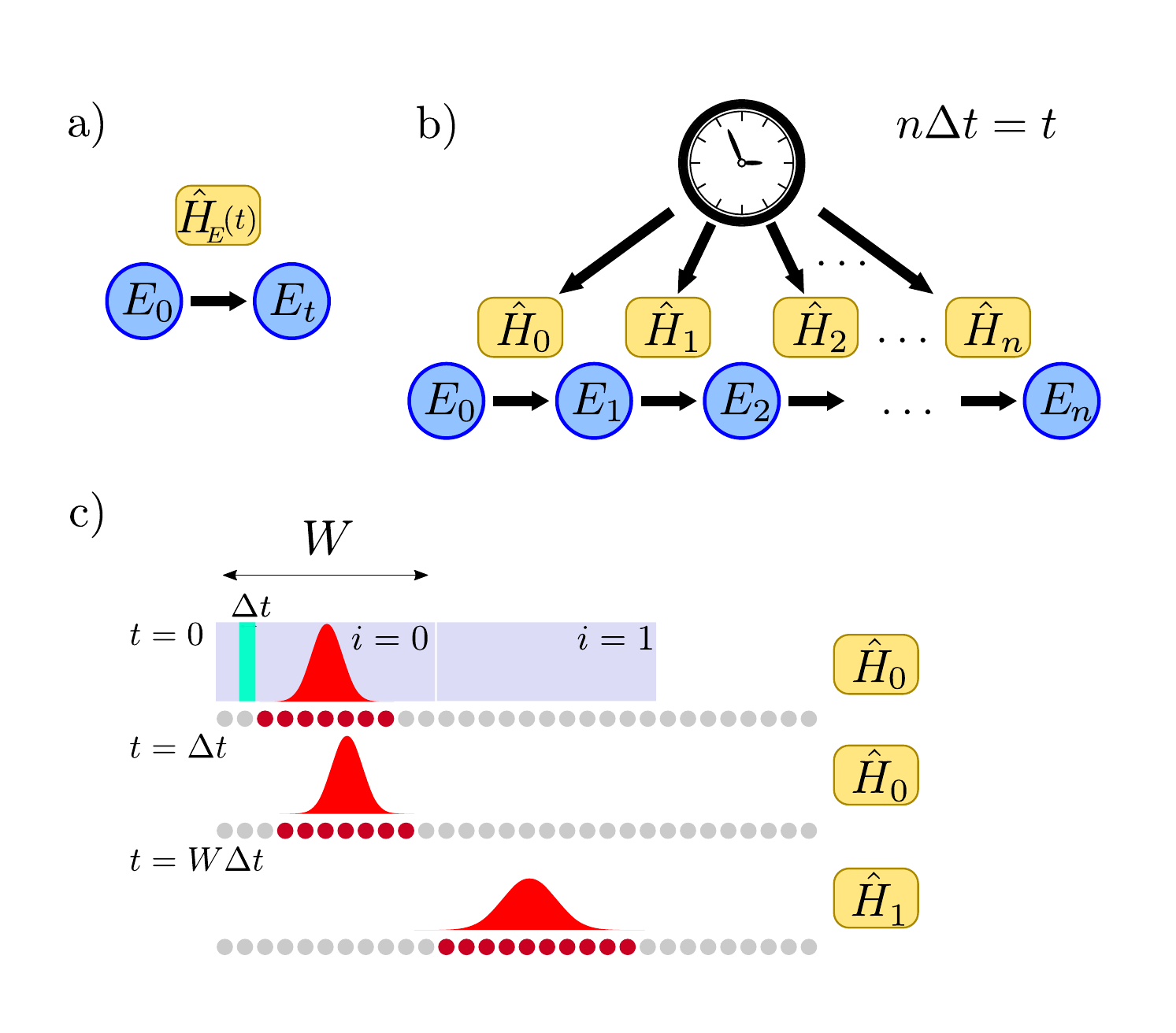}
\centering
\caption{Autonomous clock-driven quantum machine. 
a) The engine evolves from the state $E_0$ to the state $E_t$ under the action of a propagator generated by the time-dependent Hamiltonian $\hat{H}_E(t)$.
b) The time-dependent Hamiltonian is represented by a series of time-independent Hamiltonians, with the addition of a clock mechanism that chooses the appropriate generator at each intermediate stage of the time evolution. Together, the engine and clock form an autonomous system.
c) A gaussian pulse traveling on a line acts as clock by selecting the generator of the evolution of the engine.
The switching process relies on the field's intensity within a time interval of magnitude $W$, which serves as a measurement of the pulse's ``position". The field is decomposed into a succession of noninteracting time-bin ancillae. For the sake of simplicity we render the clock degradation process as a pure broadening of the pulse.}
\label{fig_0}
\end{figure}

A \textit{heat engine} can generally be defined as an open system coupled to many portions of the surrounding environment, which act as reservoirs. The primary aim of a heat engine is generating power in the form of mechanical work \cite{ReviewEngines}.
The realization of thermodynamic cycles, whether classical or quantum, relies on time-dependent Hamiltonians, which often incur significant energy costs that are difficult to fully capture in theoretical models. This problem can be addressed by embedding the system within an expanded framework that includes an additional ``clock" system to keep track of time. The engine and clock together evolve under a time-independent Hamiltonian in a larger Hilbert space \cite{Ah_Boh}. While this approach increases the dimensionality and complexity of the system due to the expanded Hilbert space, it compensates by making the overall dynamics autonomous, removing the need for external time-dependent driving.
This framework has offered valuable insights into how the intrinsic physics of the time-keeping system \cite{Woods_2018, PhysRevX.13.011016, Huber24} and the process of time sampling \cite{Meier24, PRXErker, Xuereb23} affect the dynamics of generic quantum machines, with notable implications for quantum computing protocols \cite{feynmansimulating,Feynman85,Lloyd96,watkins,tamascelli,woods2024}. Additionally, this perspective has become foundational in the resource theory of thermodynamics \cite{horodecki2013fundamental}.
In the most general framework, a quantum system that can either access or generate a time signal and use it to control the evolution of another system, while remaining robust to perturbations (such as back reactions from the engine), can serve as a clock \cite{Woods_2018}.
A minimal model including effectively all these features has been recently proposed in \cite{Malabarba_2015}: 
the time-dependent Hamiltonian coupling system and environment is modeled as a set of time-independent operators acting sequentially during the evolution time.
Each operator is correlated to the position of an external particle (the clock) freely moving under a Hamiltonian that is assumed to be linear in the momentum. 
The control mechanism is provided by an effective coupling between the particle and the engine that is assumed to commute with the free Hamiltonian of the engine (covariant operation).
As pointed out in \cite{PhysRevX.13.011016}, this model is inherently nonphysical because the particular free evolution assumed for the clock requires an unbounded clock Hamiltonian.
This argument can be seen as another facet of the celebrated \textit{Pauli objection to the existence of a time operator in standard quantum mechanics} \cite{PauliBook, HILGEVOORD200529}. Namely, the equation of motion of a self-adjoint time operator $\hat{T}$ reads
\begin{align}
    \dot{\hat{T}} = i [\hat{H},\hat{T}] = \mathbb{1}\, \Rightarrow  [\hat{T},\hat{H}] = i \mathbb{1}\,,
    \label{eq:pauli_argument}
\end{align}
thus if such operators were to exist, they would need to be unitarily equivalent to $\hat{x}$ and $\hat{p}$, implying that their spectra are continuous and unbounded from below, implying the lack of a ground state, and equilibrium. As a result, the system would require infinite energy, being unsuitable for describing physically realizable systems. In essence, the commutation relation involving $\hat{T}$ and the Hamiltonian operator, as outlined in \eq\eqref{eq:pauli_argument}, lacks \textit{exact} physical solutions within the framework of standard non-relativistic quantum mechanics.
This issue has deep historical roots within the genesis of quantum theory \cite{HILGEVOORD200529}.
Over the years, various strategies have emerged to confront this challenge, each presenting distinctive viewpoints. These strategies  include the proposition of non-self-adjoint time operators \cite{olkhovsky1974time}, as well as explorations within the framework of relativistic quantum field theory \cite{bauer2014dynamical}. Remarkably, an alternative perspective has advocated for the retention of unbounded operators \cite{Leon_2017}.
Other approaches in the literature, as discussed in \cite{Woods_2018,PhysRevX.13.011016}, often rely on well-motivated finite-dimensional approximations of the time operator $\hat{T}$.

In this work, we replace the particle used as a clock in \cite{Malabarba_2015} with a fully quantum object—an excitation of a quantum field propagating along a straight trajectory—and track time by monitoring its mean position. Keeping the structure of the coupling Hamiltonian between the system and the engine unchanged, we demonstrate that this interaction directly affects the wave packet's dispersion law, leading to a significant deterioration in the clock's performance. This result provides a framework for approximating the solution of \eq\eqref{eq:pauli_argument} using infinite-dimensional operators, which play a fundamental role in the model. More importantly, we establish a fundamental lower bound on clock degradation, showing that as long as the clock remains coupled to an engine, a certain level of deterioration is unavoidable.


The paper is structured as follows: in Sec.~\ref{sec:definition} we define the global model of engine and clock and describe the dynamics of the system using the framework of quantum collision models (QCM) \cite{Ciccarello,Gross,Ciccarello_2022,lacroix2024makingquantumcollisionmodels}. Subsequently, in Sec.~\ref{sec:int_model} we introduce a model for the interaction between the control system and the engine. This model, designed for situations characterized by minor degradation, serves to restrict deviations from the desired dynamics we aim to implement. Sec.~\ref{sec:degr} explicitly addresses the issue of degradation and in Sec.~\ref{sec:time_ops} we return to the problem of time operators to quantitatively assess how the unavoidable degradation of physical clocks affects the precise definition of time operators in QM. 

\section{Definition of engine and clock}
\label{sec:definition}
Following the reasoning in \cite{Malabarba_2015} we consider a bipartite Hilbert space $\mathcal{H} = \mathcal{H}_E \otimes \mathcal{H}_C$, where $E$ identifies a generic open system and its surrounding environment, and $C$ an auxiliary system denoted as \textit{clock}.
The role of the clock is controlling the evolution of the system without any external control.
The total Hamiltonian reads
\begin{align}
\hat{H} = \hat{H}_E + \hat{H}_C + \hat{V}_{EC} \, , 
\label{totalH}
\end{align}
where $H_E$ and $H_C$ are the free Hamiltonians of engine and clock, respectively. The identities over the complementary Hilbert spaces are omitted. In \cite{Malabarba_2015} the interaction between engine and clock takes the form
\begin{align}
\hat{V}_{EC} = \int_{\mathbb{R}} \!\!\! d x \, \hat{V}_{E} (x)\otimes \ket{x}\!_C\!\bra{x} \, ,
\label{malabarba}
\end{align}
where $\ket{x}\!_C\!\bra{x}$ projects over the different position states of the clock and $ \hat{V}_{E} \in \mathcal{H}_E$. We choose as a clock a one-dimensional bosonic field with free Hamiltonian 
\begin{align}
\hat{H}_C = \int \! d k \, \omega_k \hat{a}_k^\dag \hat{a}_k \, .
\label{eq:freeH}
\end{align}
In the position space, the initial state of the clock is a pulse described by 
\begin{align}
\ket{\psi_C}^{(0)} = \int dx \, \xi_{x_0}(x) \, \mathcal{X}(\Aod{x}) \, \ket{\rm{vac}} \equiv \hat{\psi}^{\dag}_C (x_0) \ket{\rm{vac}}\,,
\label{eq:initial_clock}
\end{align}
where $\xi_{x_0}(x)$ is the pulse envelope function centered on $x_0$, with velocity $c=1$ and $\mathcal{X}$ is an operator acting on the Fourier-transformed ladder operators of the field $\Aod{x}$ (e.g.~displacement or squeezing). We assume that the quasimonochromatic approximation \cite{PhysRevA.86.013811}, i.e.~the spectral width is much smaller than the average frequency of the wave-packet, holds. 
In the following, we will assume a linear dispersion relation $\omega_k = v_g k$ in the free Hamiltonian \eq\eqref{eq:freeH}. These assumptions correspond to taking as our clock a single-photon pulse traveling in a vacuum or a particle with a very narrow distribution in momentum.
Additionally, we assume a Gaussian envelope with frequency bandwidth $\Omega$
\begin{align}
\xi(x) = \left(\frac{\Omega^2}{2\pi}\right)^{1/4} e^{-\Omega^2 (x-x_0)^2/4}\,,
\label{gaussian_shape}
\end{align}
with $\mathcal{X} = \mathcal{I}$, where $\mathcal{I}$ is the identity map, i.e.~a single-particle pulse.

We can distinguish between two different times only when the pulse has moved by an amount that is larger than its width (in \fig\ref{fig_0} there's an example of this process). 
This leads naturally to a first discretisation of the interaction defined by a time window of size $W$ and discretising $\hat{V}_E(s)$ on them.
In light of this intuition we make the following replacements in \eq\eqref{malabarba}:
\begin{align}
\hat{V}_E (s) &\rightarrow \sum_i \hat{V}_E(s_i) \Theta_W(s-s_i)\,, \label{eq:wpot}\\
|s \ra_C\la s | &\rightarrow \frac{1}{W}\int_{\mathcal{D}_s} dx ~ \Aod{x} \Ao{x}\,,\label{eq:wproj}
\end{align}
where $\Theta(s-s_i)$ is the rectangle function of size $W$ centered on the time $s_i$, which can be chosen among the times within the $i$th window. We denote with $\mathcal{D}_q$ the integration range corresponding to a time window centered on the value $q$.
Let's assume that $s_i$ corresponds to the centre of the window.  
In the second equation the measure of position is defined, i.e.~we integrate the number of excitations per length over the region corresponding to $s$.
More detailed information about $W$ and the discrete potential $\hat{V}(s_i)$ will come in subsequent sections. 

Plugging Eqs. \eqref{eq:wpot} and \eqref{eq:wproj} into Eq.\eqref{malabarba}, the interaction term can be expressed as
\begin{align}
\hat{V}_{EC} \rightarrow &\int_{\mathbb{R}} ds \sum_i \hat{V}_e(s_i) \Theta_W(s-s_i) \otimes\frac{1}{W} \int_{\mathcal{D}_s} dx ~ \Aod{x} \Ao{x} \notag\\&=
\frac{1}{W} \sum_i   \hat{V}_e(s_i) \otimes  \int_{\mathbb{R}} ds \Theta_W(s-s_i) \int_{\mathcal{D}_s} dx ~ \Aod{x} \Ao{x} \notag\\&=
\frac{1}{W} \sum_i   \hat{V}_e(s_i) \otimes  \int_{\mathcal{D}_i} dx ~ \Aod{x} \Ao{x}  
\,,
\label{eq:hamiltonian_pre_int}
\end{align}
where $\mathcal{D}_i$ is now the region including $s_i$ (following the previous assumption it is $[s_i-W/2,s_i+W/2]$).

In the interaction picture with respect to the free Hamiltonian of the clock and the engine, we have
\begin{align}
\hat{V}_{EC}^I(t)  =
\frac{1}{W} \sum_i   \hat{V}^I_E(s_i) \otimes  \int_{\mathcal{D}_i} dx ~ \Aod{x-t} \Ao{x-t} 
\,,
\label{coupling}
\end{align}
with $\hat{V}^I_E(s_i)=e^{i \hat{H}_E t} \hat{V}_E(s_i) e^{-i \hat{H}_E t}$ (we omit the time dependence)
and the total propagator is given by
\begin{align}
\hat{\mathcal{U}}_{t,0} = e^{-i (\hat{H}_E + \hat{H}_C) t } \mathcal{T}\exp\left \{-i \int_{0}^{t}~d\tau \hat{V}^I_{EC}(\tau) \right\}\,, 
\label{prop}
\end{align}
where $\mathcal{T}$ is the time ordering operator. 

We now discretise the time axis into shorter intervals $t = n \Delta t$ with $n\in\mathbb{N}$ and the time step $\Delta t\ll W$. Thus the propagator in \eq\eqref{prop} is decomposed as 
\begin{align}
\hat{\mathcal{U}}_{t,0} &= \prod_n \hat{\mathcal{U}}_n \notag \\&= \prod_n e^{-i (\hat{H}_E + \hat{H}_C) \Delta t } \mathcal{T}\exp\left \{-i \int_{t_{n-1}}^{t_{n}}~d\tau \hat{V}^I_{EC}(\tau) \right\}\,.
\label{prop_n}
\end{align}
In the limit of $\Delta t \rightarrow 0$
we can approximate the propagator \eqref{prop_n} through the first order of the Magnus expansion \cite{Magnus1954}:

\begin{align}
\mathcal{H}_n^{(0)} &= \frac{1}{\Delta t} \frac{1}{W} \sum_i   \hat{V}_e^I(s_i) \otimes  \int_{\mathcal{D}_i} dx ~  \int_{t_{n-1}}^{t_{n}} ~d\tau \Aod{x-\tau} \Ao{x-\tau}
\notag\\&=
\frac{1}{W} \sum_i  \hat{V}^I_e(s_i) \otimes  \int_{t_{i}-W/2-t_n}^{t_{i}+W/2-t_n} dx ~ \delta \hat{n}_{x}
\notag
\\&=
\frac{1}{W} \sum_i  \hat{V}^I_e(s_i) \otimes \hat{N}_{W(i+1/2) -n} \,,
\end{align}
where $\delta\hat{n}_\alpha:=\frac{1}{\Delta t} \int_{t_{n-1}}^{t_{n}} ~d\tau \Aod{\alpha-\tau} \Ao{\alpha-\tau}$ and $\hat{N}_{W(i+1/2) -n} := \int_{t_{i}-W/2-t_n}^{t_{i}+W/2-t_n} dx ~ \delta \hat{n}_{x}$.
Thus we obtain the discrete-time propagator
\begin{align}
\hat{\mathcal{U}}_n \simeq e^{-i (\hat{H}_e + \hat{H}_c) \Delta t } \exp\left\{-\frac{i \Delta t }{W} \sum_i  \hat{V}^I_e(s_i) \otimes \hat{N}_{W(i+1/2) -n} \right\}\,.
\label{final_prop}
\end{align}

In this picture, the time evolution of the joint engine-clock system can be interpreted as a sequence of repeated interactions between the engine and a set of localized spatial modes lying within a window $W$. These modes effectively act as ancillae in a QCM, where each interaction momentarily couples the engine to a different part of the environment, thereby inducing a time-dependent evolution governed by the particle's position along the line.

Note that in our approximation we used the operator in Eq. \eqref{eq:wproj} instead of the field operator $\hat{\psi}^{\dag}_c (s) \hat{\psi}_c (s)$ to avoid biasing by the shape of the pulse.
None of the previous assumptions prevent pulse broadening, as it depends on the clock-engine interaction Hamiltonian. Note that unlike the ideal scenario proposed in \cite{Malabarba_2015}, we cannot decompose the propagator in Eq.\eqref{final_prop} into two separate operators for the clock and the engine. Consequently, the engine and the clock will exhibit correlations that grow over time. This back-reaction affects the clock's state and provides the physical basis for degradation.

\section{Clock degradation and dispersion law}
\label{sec:degr}
For linear dispersion relations $\omega_k =c k$, the free evolution of the clock ladder operators in the frequency domain reads 
\begin{align}
e^{i \hat{H}_C  t} \hat{a}_k e^{-i \hat{H}_C t} = \hat{a}_k e^{-i k \chi {\rm \sgn}(k)}\,,
\label{trasf_ops}
\end{align}
where $\chi=c t$.
As a consequence, for Gaussian envelope and under the quasimonochromatic approximation, the initial state of the clock \eqref{eq:initial_clock} transforms as
\begin{align}
\ket{\psi_C}^{(t)}
&= \int dx \int dk \, \xi_{k-k_0} e^{-i k (x-\chi {\rm \sgn}(k))} \Aod{x} \, \ket{\rm{vac}} \notag
\\&= \int dx\, \xi_{x-\chi} \Aod{x} \, \ket{\rm{vac}} \,,
\label{eq:evolved_clock}
\end{align}
i.e.~the free Hamiltonian of the clock only translates the state of the field.
It becomes evident that when the clock evolves exclusively under this Hamiltonian with the assumption of a linear dispersion, it can behave as an ideal clock in the sense of \cite{Woods_2018}. Consequently, we can explore the potential for degradation by examining the emergence of nonlinearity in the dispersion law. Starting from \eq\eqref{eq:hamiltonian_pre_int} we have, in the Fourier space,
\begin{align}
\hat{H}_C + \hat{V}_{EC} &=  \int dk v_g k \hat{a}^\dag_{k} \hat{a}_{k}  \notag
\\& 
+  \int dk dk' \mathcal{F}[\hat{V}_E](k-k')  \otimes \hat{a}^\dag_{k} \hat{a}_{k'} \,,
\label{eq:effective_int_clock}
\end{align}
where $\mathcal{F}[\hat{V}_E](k) = \int dx \hat{V}_E(x) e^{- i k x} $ denotes the Fourier transform of the engine Hamiltonian.
Even in the case in which $\la\mathcal{F}[\hat{V}_E]\ra_E \propto f(k,t) \delta(k-k')$, the interaction term still introduces a complex and non-trivial deviation from linearity into the dispersion relation. One of the consequences of the non-linearity of the dispersion relation is wave-packet broadening \cite{DIELS20061} \footnote{We are not considering Airy wavepackets \cite{berry1979nonspreading}, which are well-known for being the only nonspreading-wavepacket solutions to Schrödinger's equation. Nonetheless, it's important to note that Airy wavepackets exhibit acceleration, making them unsuitable as candidates for position-based timekeeping.}
.
When the wave-packet width becomes comparable to $W$, from the point of view of the engine's dynamics, we are not able to select which of the transformations is the one we must implement to reproduce the target dynamics: in other words, \textit{we do not know what time it is}. We interpret this phenomenon as a manifestation of clock degradation arising from the interaction with the engine. This interaction leads to the emergence of an effective mass $\propto (\partial^2_k \omega_k)^{-1}$, even when the free clock field is originally massless. On the other hand, the coupling \eqref{eq:hamiltonian_pre_int} is a translation of \eqref{malabarba} for quantum fields, which in turn finds its origins in the broader concept of the ancilla-clock system initially introduced in \cite{horodecki2013fundamental}. Therefore, for clocks falling within this broad class, our findings provide direct evidence of the intrinsic connection between degradation and the effective mass of the system operating as a clock. An interesting possibility would be preparing a pulse with an initial negative effective mass, for example by engineering the dispersion to mimic that of a particle in an inverted harmonic potential. The idea would be to counteract the induced degradation, interpreted here as the pulse acquiring a positive effective mass. The implications of such a scenario will be discussed in the next section.

\section{Time operators}
\label{sec:time_ops}
In this section, we introduce a time operator formalism that allows for a comprehensive characterization of the clock’s performance.
Specifically, our goal is to describe clock degradation at a fundamental level—namely, as the combined effect of two main contributions: the intrinsic non-idealities of the system used as the clock (referred to as \textit{self-degradation}), and the correlations arising from its interaction with an engine.
These two contributions, which we analyze in this section, form the fundamental ground upon which any additional error, such those arising from deviations from the Gaussian profile of the wave packet or the choice of the time window $W$, must be added.

Inspired by the functional form of the Hamiltonian in Eq.\eqref{eq:hamiltonian_pre_int}, we define a continuous clock time operator 
\begin{align}
\hat{T} = \frac{1}{c} \int dx \, x \, \hat{a}_x^\dag \hat{a}_x\,,
\end{align}
where $c$ is the speed of light in the vacuum. It is interesting to note that given any continuous dispersion relation $\omega(k)$ for the clock system, we can find an expression for the corresponding \eq\eqref{eq:pauli_argument}, as shown by the following
\begin{theorem}
    Given the time operator defined above, and given a continuous dispersion relation $\omega(k)$, the commutator with $\hat{H}_C$ reads
    $$ [\hat{T},\hat{H}_C] = i \hat{\Lambda} $$
    where 
    \begin{equation}
        \hat{\Lambda}=\frac{1}{c}\int dk\, \omega'(k)  \hat{a}_k^\dagger \hat{a}_k\,, 
    \end{equation} 
    with $\omega'(k) = d\omega(k)/d k$. 
\end{theorem}
\begin{proof}
    See Appendix \ref{app:time}.
\end{proof}
The theorem above implies the Heisenberg equation of motion
\begin{align}
  \dot{\hat{T}}_H(t) = i [\hat{T}_H(t),\hat{H}_C] = \int dk \, \omega'(k)\hat{a}_k^\dag \hat{a}_k = \hat{\Lambda}\,,
\end{align}
which is generally far off from the ideal case of \eq\eqref{eq:pauli_argument}, and the clock states will in general spread, a phenomenon that causes superpositions of different time states and reduces the resourcefulness of the clock system. Having introduced a time operator, we are in the position of defining an associated quantifier of degradation by means of the variance of the operator $\hat{T}$:
\begin{definition}
\label{def:degradation}
    Given a time operator $\hat{T}$ and an initial clock state $\ket{\psi_C}$, the degradation $D(t)$ of the clock is defined as 
    $$ D(t) = \sqrt{\text{Var}(\hat{T})} = \left(\langle \hat{T}^2\rangle - \langle \hat{T} \rangle ^2\right)^{\frac{1}{2}} \,. $$
\end{definition}
Given a dispersion relation and its corresponding modified Heisenberg algebra $[\hat{T},\hat{H}_C]=i\hat{\Lambda}$, we find that under its own clock Hamiltonian, the time dependency of $D(t)$ can be extracted: it is linear for any (continuous) choice of dispersion relation $\omega(k)$, and proportional to the variance of the operator $\hat{\Lambda}$:
\begin{theorem}[Clock self-degradation]
\label{th:self}
Given a dispersion relation $\omega(k)$, and its corresponding modified commutator $[\hat{T},\hat{H}_C]=i\hat{\Lambda}$, the clock self-degradation obeys
    $$  D(t) = \sqrt{\rm{Var} (\hat{\Lambda})} ~t$$
\end{theorem}
\begin{proof}
    See Appendix \ref{app:time}.
\end{proof}
The theorem above suggests that it might always be possible to completely eliminate the problem of self-degradation at all future times by preparing initial clock states such that $\text{Var}(\hat{\Lambda})=0$, i.e. eigenstates of $\hat{\Lambda}$. Interestingly enough, under fairly general assumptions this possibility is ruled out, except in the case of strictly linear dispersion relation, as shown by the following
\begin{theorem}
    Given a single-particle clock state $\ket{\psi_C}$, whose support in $k$-space is defined by a compact interval $\Omega\in\mathbb{R}$ with non-zero length, and given a continuous and injective dispersion relation $\omega(k):\Omega\to\mathbb{R}$, one has $D(t)=0$ $\forall t$ if and only if $\omega(k)$ is linear in the whole support $\Omega$.
\end{theorem}
\begin{proof}
    See Appendix \ref{app:time}.   
\end{proof}
As a corollary to the theorem above, we can consider the limit case in which our initial clock state is extremely well localised in a spatial window of width $W$. Since good clock states must be sufficiently localised in position, 
they must necessarily be sufficiently delocalised in momentum and therefore the support $\Omega$ in the theorem above becomes the whole real line as $W\to 0$, forcing the dispersion relation to be linear everywhere. \\
We can now investigate what happens to \eq\eqref{eq:pauli_argument} in the case of a linear dispersion relation $\omega(k)=ck$. Following the results above one has
\begin{align}
    \hat{\Lambda}=\int dk \, \hat{a}_k^\dag \hat{a}_k=\hat{N}_C\,,
\end{align}
where $\hat{N}_C$ counts the excitations of the field. Remarkably, this is the unique case in which we can construct resourceful clock pulses that have $\text{Var}(\hat{\Lambda})=0$ but are not eigenstates of the clock Hamiltonian (this renders them ineffective as clocks due to their stationarity).
The initial state of the clock \eq\eqref{eq:initial_clock} has a fixed number of excitations $N$, and furthermore such number of excitations is exactly conserved during the dynamics due to the fact that $[\hat{N}_C,\hat{H}]=0$. Thus, within our assumptions, when starting from a clock state with well-defined particle number $\hat{N}_C\ket{\psi_C}=N\ket{\psi_C}$, we can always define a projector $\hat{\mathbb{P}}_N$ onto the $N$-particle sector and the rescaled time operator 
\begin{equation}
    \hat{\tau} = \frac{1}{N}\hat{T} 
\end{equation}
such that
\begin{align}
 \hat{\mathbb{P}}_N(\dot{\hat{\tau}} -\mathbb{1}_N )\hat{\mathbb{P}}_N = 0\,,
\end{align}
that corresponds to the original Pauli relation in the $N$ excitation subspace \footnote{The rescaling procedure outlined here has a classical counterpart. The angular velocity of the arm is the same in all the classical watches and is achieved by rescaling the tangential speed of the arm's tip with the length of the arm.}.
In other words, we found that \textit{any} wave-packet with a fixed total number of excitations traveling through a medium with linear dispersion law (i.e.~any massless wavepacket) behaves as an ideal clock, in the sense of \cite{Malabarba_2015}. 

Now, let us assume that our clock is coupled to an engine $E$. Contrary to what we had before, the clock now does not evolve unitarily, but as a sum of unitaries: 
\begin{lemma}
    \label{lem:0}
    The state of the clock at a time $t$, when coupled to an engine is given by 
    $$\rho_C(t) = {\rm Tr}\rho_{CE}(t) = \sum_j p_j \hat{\mathcal{U}}_j |\psi_C\rangle\langle\psi_C| \hat{\mathcal{U}}_j^\dagger,$$
    where $\hat{\mathcal{U}}_j = e^{-i(\hat{H}_C +  g\hat{W}^{(j)})t}$ is the unitary operator that accounts for the  engine's effect on the clock Hamiltonian, represented here by $\hat{W}^{(j)}t$ and $g$ is the coupling strength between clock and engine.
\end{lemma}
\begin{proof}
    See Appendix \ref{app:time}. 
\end{proof}
\begin{lemma}
\label{lem:1}
In the open system scenario, the expectation value of the time operator acquires an additional contribution due to the interaction with the engine. Specifically, up to second order terms in $g$
    \begin{equation*}
        \langle \hat{T}\rangle = \langle \hat{\Lambda}\rangle\, t + g\,\langle \hat{V}(t)\rangle,  
    \end{equation*}
    where $\langle \hat{V}(t)\rangle = i\,\langle [ \hat{\mathcal{W}}(t), \hat{T}_H(t) ]\rangle$, and $\hat{\mathcal{W}}(t) := \int_{0}^t d\tau\, \sum_j p_j\,\,\hat{W}^{(j)}_H(-\tau)$.
\end{lemma}
\begin{proof}
    See Appendix \ref{app:time}. 
\end{proof}
\begin{theorem}[Clock degradation]
Consider a quantum clock with free Hamiltonian $H_C$, coupled to an engine with Hamiltonian $H_E$, with coupling strength $g$. For sufficiently small $g$, i.e. weak-coupling limit, up to second order terms in $g$, the degradation of the clock is
    $$D(t) = \sqrt{{\rm Var}(\hat{\Lambda})t^2 + 2 \,g \,{\rm Cov}(\hat{V}(t), \hat{T}_H(t))}.$$
\end{theorem}
\begin{proof}
    See Appendix \ref{app:time}. 
\end{proof}
As highlighted by the theorem above, eliminating clock degradation entirely is fundamentally impossible. More precisely, starting with a non-trivial dispersion relation, associated with a negative effective mass, it might be possible to end up with a linear dispersion, eliminating the contribution of the first term. The second term encodes the correlation build-up between clock and engine, and as such it cannot be compensated by this effect.
\\The trade-off between degradation and resourcefulness of the clock states is a consequence of the Heisenberg uncertainty relation for the clock's position and momentum, and can then be translated into a lower-bound: 
\begin{align}
    \Delta \hat{H}\Delta \hat{T}_H(t)&\geq \frac{1}{2}|\langle [\hat{H}, \hat{T}_H(t)]\rangle|\\
    &=\frac{1}{2}|\langle \dot{\hat{T}}_H(t)\rangle|
\end{align} 
\\Without loss of generality, we can decompose $T_H(t)$ into two parts : one that describes the non-interacting evolution of the clock and another that accounts for the deviations induced by the interaction with the engine. The time derivative of the operator then takes the form $\dot{\hat{T}}_H = \hat{\Lambda} + \delta \hat{\Lambda} (t)$, where the first and second term capture the time-independent and the time dependent corrections due to the engine, respectively. The uncertainty relation becomes 
\begin{equation}
   \Delta \hat{H} \Delta \hat{T}_H(t) \geq \frac{1}{2} |\langle \hat{\Lambda} + \delta \hat{\Lambda} (t)\rangle|. 
\end{equation}
This allows us to write a lower bound for the degradation of the clock as 
\begin{equation}
    D(t) \geq \frac{|\langle \hat{\Lambda}\rangle + \langle\delta \hat{\Lambda} (t)\rangle|}{2 \Delta \hat{H}}
\end{equation}
where $\Delta \hat{H}$ is the variance of the total Hamiltonian $\hat{H}$, specified in Eq.\eqref{totalH}. This follows directly from the general Robertson-Schrödinger uncertainty relation associated to the pair of noncommuting observables $\hat{T}$ and $\hat{H}$ \cite{PhysRev.34.163}.
Considering an explicit dependence with the engine Hamiltonian, it follows that $\langle T_H(t) \rangle = \langle T \rangle$, and from Lemma \ref{lem:1} we have 
\begin{equation}
    \langle \delta\hat{\Lambda}(t) \rangle = g\,\langle \dot{\hat{V}}(t) \rangle,  
\end{equation}
where $\langle \dot{\hat{V}}(t) \rangle = \frac{d}{dt}\langle \hat{V}(t) \rangle$. The uncertainty relation then becomes 
\begin{equation}
    D(t) \geq \frac{|\langle \hat{\Lambda}\rangle + g\, \langle \dot{\hat{V}}(t)\rangle|}{2 \Delta \hat{H}}.  
\end{equation}
The bound quantifies a fundamental limitation of the degradation of the clock. Note that $\langle \hat{\Lambda}\rangle$ characterizes the linear ticking of the clock in the isolated scenario, while $\langle \dot{\hat{V}}(t)\rangle$ accounts for deviations induced by the interaction with the engine that manifests as an additional source of degradation. This term captures the fact that using a clock comes with a thermodynamic price, correlations and backaction disturb the clock's evolution. The engine-induced dynamics introduce nonlinearity into the clock’s evolution, effectively reducing its temporal precision. However, the energy uncertainty sets the scale of this trade-off. Note that $\Delta \hat{H}$ encodes the variance in energy, which set a fundamental constraint on how precisely the clock can resolve time. A higher energy uncertainty allows, in principle, for finer temporal resolution, thereby tolerating greater dynamical changes before significant degradation occurs. It is important to note that the term $\langle \dot{\hat{V}}(t)\rangle$ need not be strictly positive. In certain regimes, the interaction with the engine may transiently reduce the degradation rate, or even partially reverse it. This could occur, for example, if the coupling induces correlations or coherence that stabilize the clock’s evolution. However, such scenarios are typically non-generic and require an exotic fine-tuned dynamics. 

\section{Conclusions}

Our formalism, built upon a minimal yet well-defined and comprehensive physical framework, proves effective in capturing the complex nature of timekeeping and the deviations of a clock's behavior from ideal performance. A noteworthy connection to \cite{Woods_2018} arises when considering the dimensionality of the clock system. In our approach, the clock is modeled as a field, i.e., a continuous-variable system with an infinite-dimensional Hilbert space as support. This naturally rises the expectation that such clocks might exhibit near-ideal behavior, especially given that the errors discussed in \cite{Woods_2018} decay exponentially with the Hilbert space dimension $d$. However, a crucial assumption in those bounds is that the clock is prepared in a state with an energy spread $\sigma\sim\sqrt{d}$ and a mean energy standing at the midpoint of the energy spectrum. For a harmonic oscillator, satisfying this condition would imply a state of infinite energy, thereby departing from the physically grounded conditions we consider here.
We identify a fundamental form of clock degradation—that is, independent of practical limitations such as the preparation of high-fidelity clock states or the implementation of a well-defined clock-position measurement protocol—and decompose it into two primary contributions.
The first one, the self-degradation, arises from the clock’s own Hamiltonian dynamics and manifests as pulse spreading, allowing us to draw a clear connection between the degradation of a quantum clock and its mass in the low-energy regime. This effect can be mitigated by maintaining a linear dispersion relation in the medium, where rescaling mechanisms act to prevent the degradation from becoming unbounded. Notably, the link between a clock’s mass and its capacity to keep time has also been explored experimentally \cite{clock_mass_exp}.
The second contribution to degradation is induced by the interaction between the clock and the engine, and is strongly influenced by the correlations between them. Together, these effects give rise to a non-trivial connection between timekeeping precision and the system’s energy uncertainty. This is captured by an uncertainty relation that imposes a fundamental lower bound on degradation, thereby setting a limit to the achievable precision of real clocks.
Interestingly, the interplay between these two fundamental sources of degradation may give rise to unconventional behaviors, including the possibility of suppressing or even reversing degradation—phenomena typically associated with nonlinear systems. 
Our formalism could be used in the future to explore such scenarios, including moving-particle clocks within the framework of nonlinear quantum mechanics, such as soliton-based models.
These systems may uncover regimes where degradation is not only reduced, but actively controlled or reversed, opening up new possibilities for improving timekeeping precision in quantum technologies.
\\

\section*{Acknowledgements}
We thank M. P. Woods for his detailed comments and the many fruitful discussions that followed.
This work was supported by the BMBF project PhoQuant (Grant No. 13N16110) and the Quantera project ExTRaQT (Grant No. 499241080).

\appendix

\section{On the form of the interaction}
\label{sec:int_model}
The choice of the window $W$ and the interaction is crucial for the scenario we want to describe. We operate under the assumption that the interaction in \eq\eqref{coupling} does not allow energy exchange between the engine and the clock. Thus apart from translation of the clock states, the pulse may suffer \textit{degradation}, i.e. pulse broadening in time. The most critical parameter for this timekeeping approach based on the position is then the width of the wave-packet. 
One possibility can be choosing $W$ in order to include all the pulse envelope ($ W \gg \Omega^{-1}$ for the Gaussian pulse \eq\eqref{gaussian_shape}, as depicted in \fig\ref{fig_0})
and
\begin{align}
\hat{V}^I_e(s_i)=
\begin{cases}
\hat{v}_i \, ,{\rm for}~i~{\rm even}\,;\\
\mathbb{1}\,,{\rm for}~i~{\rm odd}\,.
\end{cases}
\label{cond}
\end{align}
This choice guarantees that a specific Hamiltonian $v_i$ acts on the engine at a specified time window, at the price of a phase shift due to the identity. 
The presence of intermittent \textit{dead windows} naturally introduces the concept of a ``period" in our system, similar to mechanical clocks: the interval between two successive ticks corresponds to a phase of the evolution during which the oscillating component disengages from the surrounding mechanism (\fig\ref{fig_2}).
More precisely, for each $\Delta t$, the potential $v_i$ undergoes modulation proportional to the radiation flux over the $i$th time window. 
In this scenario the width of the pulse in terms of $\Delta t$ is crucial for the speed at which one can induce a certain $\hat{V}_e(t)$ on the engine.
\begin{figure}[htb!]
\includegraphics[scale=0.75,angle=0]{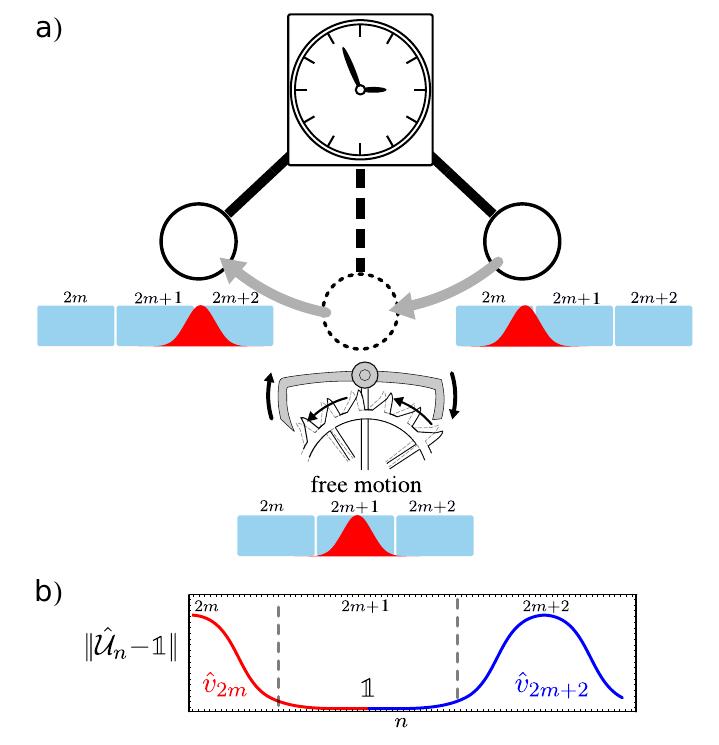}
\caption{(a) Mechanical analogue of the time windows scheme \eqref{cond}. During the time between two successive ticks, the ends of the anchor (rigidly connected to the pendulum's rod) don't push on the teeth of the escape wheel and both pendulum and mechanism move freely.
(b) Qualitative behaviour of the generator of the propagator \eqref{final_prop} under the scheme \eqref{cond}. The action of the identity during the odd time windows causes the generator undergoes a smooth transition from $\hat{v}_{2m}$ to $\hat{v}_{2m+2}$. The colors correspond to the generator that is acting in the considered time. The shape of the curve depends on the shape of the pulse's envelope.}
\label{fig_2}
\end{figure}
Let us consider $W(i+1/2) - n' \in [2m W, (2m+1)W]$, with $m \in \mathbb{N}$. In this case the pulse encompasses two time windows. According to \eqref{cond} and considering the action of $\hat{N}_\alpha$ on the state of the clock, the generator reads
\begin{align}
\frac{1}{W} \sum_i  &\hat{V}^I_e(s_i) \otimes \hat{N}_{W(i+1/2) -n'}  =\alpha_{n'}\hat{V}^I_e(s_{2m}) + \beta_{n'} \hat{V}^I_e(s_{2m+1})\notag
\\& =\alpha_{n'} \hat{v}_{2m} + \beta_{n'} \mathbb{1}\,,
\end{align}
where we define $\alpha_{n'} = \frac{1}{W}\la \hat{N}_{W(2m+1/2) -n'} \ra_c^{1/2}$ and $\beta_{n'} = \frac{1}{W}\la \hat{N}_{W(2m+1+1/2) -n'} \ra^{1/2}_c$. Thus, in this case,
\begin{align}
\hat{\mathcal{U}}_{n'} &= \exp\{-i (\alpha_{n'} \hat{v}_{2m} + \beta_{n'} \mathbb{1})\Delta t\} \notag\\&= \exp\{-i \alpha_{n'} \hat{v}_{2m} \Delta t\} \exp\{-i\beta_{n'} \mathbb{1}\Delta t\}\notag
\\&= \exp\{-i \alpha_{n'} \hat{v}_{2m} \Delta t\} e^{-i\beta_{n'} \Delta t}\,,
\end{align}
and the global phase factor can be neglected.
At each $n'$ the potential $\hat{v}_{2m}$ is weighted by the factor $\alpha_{n'}$ that ranges between $0$ (the pulse is crossing another window) and a maximum that depends on the normalisation of the pulse (when the pulse is exactly centered on the $2m$th time window).
This scheme allows us to approximate the action of the target time-dependent generator $V_e(t)$ on a mesh of time windows of size $W$. Nevertheless, a notable limitation is that each generator $V_e(s_i)$ demands an effective time frame of $2W$ for execution. This, in conjunction with the modulation, imposes significant constraints on the accuracy of the achieved transformation. Notably, as $W$ increases, our discrete series offers a closer approximation to the continuous generator. Such trade-offs, that arise from the need to bridge the gap between continuous and discrete representations, are a common challenge in the realm of approximating continuous-time systems.
Up to this point, we have not directly accounted for any clock degradation process. In the upcoming section, we explore the impact of clock degradation on the wave-packet dynamics, examining how the Hamiltonian in \eq\eqref{eq:hamiltonian_pre_int} affects also its shape. Understanding the underlying reasons for this degradation is crucial, given our reliance on the position of the particle in our time-keeping protocol.

\section{Proofs of the statements in section VI}
\label{app:time}
For the sake of clarity, in this appendix we will focus exclusively on the dynamics of the clock setting aside the explicit goal of implementing a particular transformation. Thus we consider $V(t)$ a continuous variable as well as in \cite{Malabarba_2015}, i.e.~we work in the limit of $W\rightarrow 0$.
Therefore the clock Hamiltonian keeps the form in \eq\eqref{eq:freeH} while the total Hamiltonian is turned into

\begin{equation}
    \hat{H} = \hat{H}_E\otimes \mathbb{1}_C + \mathbb{1}_E\otimes \hat{H}_C + \int dx\, \hat{V}_E(x) \otimes \hat{a}_x^\dagger \hat{a}_x\,.
\end{equation}
When considering the time operator
\begin{equation}
    \hat{T} = \frac{1}{c}\int dx \, x\, \hat{a}_x^\dagger \hat{a}_x\,,
\end{equation}
where $c$ is a reference speed, we are interested in the commutator $[\hat{T}_H(t),\hat{H}_C]$, where $\hat{T}_H(t)$ is the Heisenberg picture representation of $\hat{T}$, which will give us the equation of motion we are looking for. In order to compute the commutator with $\hat{H}_C$, we make use of the following
\begin{lemma}[$k$-basis representation of $\hat{T}$]
    $$ \hat{T} = \frac{i}{c} \lim_{\epsilon\to 0} \frac{1}{\epsilon}\int dk\big( \hat{a}_k^\dagger \hat{a}_{k+\epsilon} - \hat{a}^\dagger_k \hat{a}_k \big)$$
\end{lemma}
\begin{proof}
We start by expressing $\hat{a}_x,\hat{a}_x^\dagger$ as anti-transforms of $\hat{a}_x,\hat{a}_x^\dagger$, and we make use of the fact that, in the sense of distributions $\int dx x e^{-ikx} = 2\pi i \delta'(k)$ where $\delta'(x)$ is such that $\int dx \delta'(x)f(x) = - f'(0)$. Putting this together we can write
\begin{equation}
\begin{split}
     & \hat{T} = \frac{1}{2\pi}\int dx \int dk \int dk' \, x e^{i (k'-k) x} \Ao{k}^\dagger \Ao{k'} \\& = \frac{1}{2\pi}\int dk \int dk' \Big(\int dx\, x e^{i (k'-k) x} \Big) \Ao{k}^\dagger \Ao{k'}\\ &  =\frac{1}{2\pi}\int dk \int dk' \Big(-2\pi i \delta'(k'-k) \Big) \Ao{k}^\dagger \Ao{k'} \\&= -i \int dk \Ao{k}^\dagger \int dk' \delta'(k'-k) \Ao{k'} = i \int dk \Ao{k}^\dagger \Big[ \frac{\partial}{\partial k'} \Ao{k'} \Big]_{k'=k}\,.
\end{split}
\end{equation} 
Finally, by using the definition of derivative
\begin{equation}
    \frac{\partial}{\partial k} \hat{a}_{k}  = \lim_{\epsilon\to 0} \frac{1}{\epsilon} \Big( \hat{a}_{k\epsilon} - \hat{a}_k \Big)\,,
\end{equation}
we obtain the result.
\end{proof}
\begin{theorem}
    Given the time operator defined above, and given a continuous dispersion relation $\omega(k)$, the commutator with $\hat{H}_C$ reads
    $$ [\hat{T},\hat{H}_C] = i \hat{\Lambda} $$
    where $\hat{\Lambda}=\frac{1}{c}\int dk\, \omega'(k) \hat{a}_k^\dagger \hat{a}_k$\,.
\end{theorem}
\begin{proof}
The quantity we need to compute is
\begin{equation}
\begin{split}
[\hat{T},\hat{H}_C] =  \frac{i}{c} \lim_{\epsilon\to 0} \frac{1}{\epsilon}\int dk \int dk' \, \omega_{k'}\Big[ \hat{a}_k^\dagger \hat{a}_{k+\epsilon} - \hat{a}^\dagger_k \hat{a}_k ,  \hat{a}_{k'}^\dagger \hat{a}_{k'} \Big]\,.
\end{split}
\end{equation}
By making use of the relations 
\begin{equation}
\begin{split}
   &  [\Ao{k}^\dagger a_{k+\epsilon}, a^\dagger_{k'}\Ao{k'}] = \Ao{k}^\dagger \Ao{k'} \delta(k+\epsilon-k') - \Ao{k'}^\dagger a_{k+\epsilon}\delta(k-k')\\
   &  [\Ao{k}^\dagger \Ao{k}, a^\dagger_{k'}\Ao{k'}] = \Big( \Ao{k}^\dagger \Ao{k'} - \Ao{k'}^\dagger \Ao{k} \Big) \delta(k-k')\,,
\end{split}
\end{equation}
we get
\begin{equation}
\begin{split}
 & [\hat{T},\hat{H}_C] =  \frac{i}{c} \lim_{\epsilon\to 0} \frac{1}{\epsilon} \Big(\int dk \, \omega_{k+\epsilon}\Ao{k}^\dagger \Ao{k+\epsilon} - \int dk\, \omega_k\Aod{k} \Ao{k+\epsilon}  \Big) \\  & = \frac{i}{c} \lim_{\epsilon\to 0} \frac{1}{\epsilon} \int dk \, (\omega_{k+\epsilon} - \omega_k )\Aod{k} \Ao{k+\epsilon} \,.
\end{split}
\end{equation}
Now, if $\omega(k) \equiv \omega_k$ is a continuous function of $k$ we can write 
\begin{equation}
\begin{split}
   & \lim_{\epsilon\to 0} \Big(\frac{\omega_{k+\epsilon} - \omega_k }{\epsilon} \Aod{k} \Ao{k+\epsilon}  \Big) \\ & =\lim_{\epsilon\to 0} \frac{\omega_{k+\epsilon} - \omega_k }{\epsilon} \times \lim_{\epsilon\to 0}\Aod{k} \Ao{k+\epsilon} = \omega'(k)\Aod{k}\Ao{k} \,,   
\end{split}
\end{equation}
and therefore
\begin{equation}
    [\hat{T},\hat{H}_C] = \frac{i}{c}\int dk\, \omega'(k)\Ao{k}^\dagger\Ao{k}
\end{equation}
\end{proof}
From the theorem above, we can easily drawn some initial conclusions, as exemplified in the following two corollaries:
\begin{corollary}
    Given the time operator defined above, and given any continuous dispersion relation 
    $\omega(k)$, 
    the resulting modified commutator
    $$ [\hat{T},\hat{H}_C] = i \hat{\Lambda} $$
    is such that 
    $[\hat{\Lambda},\hat{H}_C]=0$, 
    i.e. 
    $\hat{\Lambda}$ 
    is a constant of motion. In particular, the Heisenberg picture representation of 
    $\hat{\Lambda}$ 
    reads
    $$ \hat{\Lambda}_H(t) = e^{i\hat{H}_Ct}\hat{\Lambda} e^{-i\hat{H}_Ct} = \Lambda\,.$$
\end{corollary}

\begin{corollary}
    Given the time operator defined above, and given a linear dispersion relation $\omega(k)=ck$, the commutator with $\hat{H}_C$ reads
    $$ [\hat{T},\hat{H}_C] = i \hat{N} $$
    where $\hat{N}=\int dk\,\hat{a}_k^\dagger \hat{a}_k$ is the number operator.
\end{corollary}

As a consequence, note that the Heisenberg equations of motion for this time operator read
\begin{equation}
    \dot{\hat{T}}_H = i [\hat{H}_C,\hat{T}_H] = i[\hat{H}_C,\hat{T}] = \hat{\Lambda}\,.
\end{equation}
We are now in a position to define the concept of a clock's degradation, i.e. the spreading of a clock state under its own Hamiltonian dynamics.

\begin{definition}[Clock Degradation]
    Given a time operator $\hat{T}$ and an initial clock state $\ket{\psi_C}$, the degradation $D(t)$ of the clock is defined as 
    $$ D(t) = \sqrt{\text{Var}(\hat{T})} = \left(\langle \hat{T}^2\rangle - \langle \hat{T} \rangle ^2\right)^{\frac{1}{2}} \,. $$
\end{definition}

Given the modified Heisenberg algebra $[\hat{T},\hat{H}_C]=i\hat{\Lambda}$, we can characterize the resulting degradation of a clock state. We find that the time dependency of $D_0(t)$ can be extracted: it is linear for any (continuous) choice of dispersion relation $\omega(k)$, and proportional to the variance of the operator $\hat{\Lambda}$.

\begin{theorem}[Characterization of self-degradation]
Given a dispersion relation $\omega(k)$, and its corresponding modified commutator $[\hat{T},\hat{H}_C]=i \hat{\Lambda}$, the clock degradation obeys
    $$  D(t) = \sqrt{\rm{Var} (\hat{\Lambda})} t$$
\end{theorem}
\begin{proof}
    First, we exploit the freedom of writing expectation values in the Schr\"odinger or Heisenberg picture as follows
    \begin{equation}
        \langle \hat{T}^2\rangle = \bra{\psi_C(t)} \hat{T}^2 \ket{\psi_C(t)} = \bra{\psi_C} \hat{T}^2_H(t) \ket{\psi_C}\,.
    \end{equation}
Then, we use the Heisenberg equation of motion
\begin{equation}
    \dot{\hat{T}}_H = \hat{\Lambda}_H(t) = \hat{\Lambda}
\end{equation}
implying the solution $\hat{T}_H(t) = \hat{\Lambda} t$, which we can plug in the expression for $D(t)$ and get
\begin{equation}
    D_0(t) = \sqrt{\bra{\psi_c} \hat{\Lambda}^2 t^2\ket{\psi_c} - \bra{\psi_c}\hat{\Lambda} t \ket{\psi_c}^2}\,.
\end{equation}
By extracting the time dependency we get the result.
\end{proof}

It is interesting to note that, under fairly general assumptions, linear dispersion relations are the only ones that guarantee the existence of non-degrading clock states, as shown by the following

\begin{theorem}
    Given a single-particle clock state $\ket{\psi_C}$, whose support in $k$-space is defined by a compact interval $\Omega\in\mathbb{R}$ with non-zero length, and given a continuous and injective dispersion relation $\omega(k):\Omega\to\mathbb{R}$, one has $D(t)=0$ $\forall t$ if and only if $\omega(k)$ is linear in the whole support $\Omega$
\end{theorem}
\begin{proof}
     Let us suppose that $D(t)=0$ identically. Then $\text{Var}(\hat{\Lambda})=0$ on $\ket{\psi_C}$, which means that $\ket{\psi_C}$ is an eigenstate of $\hat{\Lambda}$. Furthermore, since $\omega(k)$ is an injective function, $\ket{\psi_C}$ cannot be an eigenstate of $\hat{H}_C$ if $\Omega$ has non-vanishing length. However, since $[\hat{\Lambda},\hat{H}_C]=0$, the only possibility is that $\hat{\Lambda}$ is degenerate on $\Omega$, i.e. the function $\omega'(k)$ is constant on $\Omega$. Therefore $\omega(k)$ is linear on $\Omega$.\\
     Conversely, let us suppose that $\omega(k)$ is linear. Then $\hat{\Lambda} = \hat{N}$ and, since $\hat{N}\ket{\psi_C}=\ket{\psi_C}$ we have $\text{Var}(\hat{N})=0$ and therefore $D(t)=0$.
\end{proof}

When a clock is coupled to an engine, we can assume without loss of generality that the initial state is initially uncorrelated and given by $|\psi\rangle^{(0)} = |\psi_{C}\rangle^{(0)} \otimes |\psi_{E}\rangle^{(0)}$, evolving under the total Hamiltonian $H = H_C + H_E + V_{EC}$.
\begin{lemma}
    The state of the clock at a time $t$, when coupled to an engine is given by 
    $$\rho_C(t) = {\rm Tr}_E\rho_{CE}(t) = \sum_j p_j\hat{\mathcal{U}}_j |\psi_C\rangle\langle\psi_C| \hat{\mathcal{U}}_j^\dagger,$$
    where $\hat{\mathcal{U}}_j = e^{-i(\hat{H}_C +  g\hat{W}^{(j)})t}$, such that $g$ encodes the coupling strength between clock and engine and $\hat{W}^{(i)}$ encompasses the contribution from the engine to the clock Hamiltonian.  
\end{lemma}
\begin{proof}
Let us write the initial state in the engine eigenbasis, $|\psi\rangle^{(0)} = |\psi_{C}\rangle^{(0)} \otimes \sum_j c_j |j\rangle$, such that $\hat{H}_E|j\rangle = E_j |j\rangle$. The evolved composed state can be obtained as 
\begin{align*}
    |\psi\rangle(t) &= e^{-i \hat{H} t}\, |\psi\rangle^{(0)}\\
    & = e^{-i (\hat{H}_C + \hat{H}_E + \hat{V}_{EC}) t}\, |\psi_C\rangle^{(0)} \otimes \sum_j c_j |j\rangle\\
    & = \sum_j c_j |j\rangle \otimes  e^{-i (\hat{H}_C + E_j+ \langle k|\hat{V}_{EC}|j\rangle) t} |\psi_C\rangle^{(0)}\\
    &= \sum_j c_j |j\rangle \otimes  e^{-i (\hat{H}_C + g\hat{W}^{(j)}) t} |\psi_C\rangle^{(0)}, 
\end{align*}
where in the second to last line we defined $g\hat{W}^{(j)} = E_j + \langle j|\hat{V}_{EC}|j\rangle$ as the effect on the clock due to the coupling with the engine, with coupling strength $g$. The reduced state of the clock at a time $t$, is obtained by simply tracing out the engine 
\begin{align*}
    \rho_C(t) &= \text{Tr}_E[\rho(t)]\\
    &= \text{Tr}_E\left[ e^{-i \hat{H} t}\, |\psi\rangle\langle\psi|^{(0)} \,e^{i \hat{H} t}\right]\\
    &= \text{Tr}_E\Big[\sum_{j'}\sum_{j''} c_{j'} c_{j''} |j'\rangle\langle j''| e^{-i (\hat{H}_C + g\hat{W}^{(j)}) t}\\
    &\hspace{4em}|\psi_C\rangle\langle\psi_C|^{(0)}\,e^{i (\hat{H}_C + g\hat{W}^{(j)}) t}\Big]\\
    &= \sum_j|c_j|^2 e^{-i (\hat{H}_C + g\hat{W}^{(j)}) t} |\psi_C\rangle\langle\psi_C|^{(0)}\,e^{i (\hat{H}_C + g\hat{W}^{(j)}) t}\\
    &= \sum_j p_j \hat{\mathcal{U}}_j \rho_C(0)\hat{\mathcal{U}}_j^\dagger. 
\end{align*}
The clock does not evolve unitarily, but with a sum of unitaries, where each $\hat{\mathcal{U}}_j = e^{-i(\hat{H}_C + g\hat{W}^{(j)})t}$ encodes the contribution of the engine interaction. We treat $g \hat{W}^{(j)}$ as a perturbation to the free Hamiltonian of the clock.
\end{proof}

\begin{lemma}[Expectation value of the time operator]
In the open system scenario, the expectation value of the time operator acquires an additional contribution due to the interaction with the engine. Specifically,
    \begin{equation*}
        \langle \hat{T}\rangle = \langle \hat{\Lambda}\rangle\, t + g\,\langle \hat{V}(t)\rangle,  
    \end{equation*}
    where $\langle \hat{V}(t)\rangle = i\,\langle [ \hat{\mathcal{W}}(t), \hat{T}_H(t) ]\rangle$, and $\hat{\mathcal{W}}(t) := \int_{0}^t d\tau\, \sum_j p_j\,\,\hat{W}^{(j)}_H(-\tau)$.
\end{lemma}
\begin{proof}
The expectation value of the time operator in the open system scenario can be computed as
\begin{align}
    \langle \hat{T}\rangle &= \text{Tr}[\rho_C(t)\,\hat{T}] \nonumber\\
                    & = \sum_j p_j \text{Tr}[\rho_C(0)\,\hat{\mathcal{U}}_j^\dagger \,\hat{T} \, \hat{\mathcal{U}}_j].
\end{align}
Let us now compute the term on the right hand side of the equation  above 
\begin{align*}
    \hat{\mathcal{U}}_j^\dagger \,\hat{T} \, \hat{\mathcal{U}}_j &= e^{-i(\hat{H}_C + g\hat{W}^{(j)})t} \, \hat{T} \, e^{i(\hat{H}_C +  g\hat{W}^{(j)})t} \\
    &= \hat{\mathcal{V}}^{(j)}\, \hat{T}_{H}(t) (\hat{\mathcal{V}}^{(j)})^* \\
    &= \hat{T}_{H}(t) + i g \int_{0}^t d\tau \, \big( \hat{W}^{(j)}_H(-\tau)\, \hat{T}_H(t) \\
    &\hspace{4em} - \hat{T}_H(t)\,\hat{W}^{(j)}_H(-\tau) \big) \\
    &= \hat{T}_{H}(t) - i g \left[\hat{T}_H(t), \int_{0}^t d\tau \, \,\hat{W}^{(j)}_H(-\tau)\right]
\end{align*}
where $\hat{\mathcal{V}}^{(j)} =\left(\mathbb{1} + ig \int_{0}^t d\tau \, e^{-i \hat{H}_C\tau}\, \hat{W}^{(j)}\, e^{i\hat{H}_C\tau} \right)$, corresponds to an expansion of $\hat{\mathcal{U}}_j$ to first order in $g$. Then, we move to the Heisenberg picture. The expectation value of the time operator can then be written as 
\begin{equation}
\label{eq:meanT}
    \langle \hat{T}\rangle = \langle \hat{\Lambda}\rangle\, t +  i g \langle [ \hat{\mathcal{W}}(t), \hat{T}_H(t) ]\rangle
\end{equation}
where $\hat{\mathcal{W}}(t) := \int_{0}^t d\tau\, \sum_j p_j\,\,\hat{W}^{(j)}_H(-\tau)$, and the first term corresponds to the self-degradation, computed in Theorem \ref{th:self}.  It follows that $ \langle \hat{V}(t)\rangle = i\,\langle [ \hat{\mathcal{W}}(t), \hat{T}_H(t) ]\rangle$. 
\end{proof}

\begin{theorem}[Clock-degradation]
    Consider a quantum clock with free Hamiltonian $H_C$, coupled to an engine with Hamiltonian $H_E$, with coupling strength $g$. For sufficiently small $g$, i.e. weak-coupling limit, up to second order terms in $g$, the degradation of the clock is
    $$D(t) = \sqrt{{\rm Var}(\hat{\Lambda})t^2 + 2 \,g \,{\rm Cov}(\hat{V}(t), \hat{T}_H(t))}.$$
\end{theorem}
\begin{proof}
In the context of a perturbation to the free Hamiltonian, we neglect second order terms in $g$, which give us  
\begin{equation}
    \langle \hat{T} \rangle ^2 = \langle \hat{\Lambda} \rangle ^2 t^2 + 2 i g \langle \hat{\Lambda} \rangle \langle [\hat{\mathcal{W}}(t), \hat{T}_H(t)]\rangle \,t + \mathcal{O}(g^2)
\end{equation}
and 
\begin{equation}
    \langle \hat{T} ^2 \rangle = \langle \hat{\Lambda}^2 \rangle t^2 + i g  \langle \hat{\Lambda} \rangle \langle [\hat{\mathcal{W}}(t), \hat{T}_H^2(t)]\rangle + \mathcal{O}(g^2)
\end{equation}
The variance of the time operator $\hat{T}$ can be written as  
\begin{equation}
    \text{Var}(\hat{T}) = \text{Var}(\hat{\Lambda})t^2 + 2 i g \text{ Cov}([\hat{\mathcal{W}}(t), \hat{T}_H], \hat{T}_H) + \mathcal{O}(g^2)
\end{equation}
The equation above can be written as
\begin{equation}
    \label{eq:Var}
    \text{Var}(\hat{T}) = \text{Var}(\hat{\Lambda})t^2 + 2 g \text{ Cov}(\hat{V}(t), \hat{T}_H(t)) + \mathcal{O}(g^2)
\end{equation}
where, again we used $\hat{V}(t) = i [\hat{\mathcal{W}}(t), \hat{T}_H(t)]$. 
\end{proof}

\bibliography{biblio}
\end{document}